\documentclass[journal]{IEEEtran}
\usepackage{amsmath}
\usepackage{amssymb}
\usepackage{amsthm}
\usepackage{bm}
\usepackage{graphicx}
\usepackage[hyphens]{url}
\IfFileExists{cite.sty}{\usepackage{cite}}{}

\makeatletter
\g@addto@macro\UrlBreaks{%
\do\/\do-\do\.\do\_\do\?\do\&\do=\do\#\do\%\do\~}
\makeatother
\newtheorem{theorem}{Theorem}

\title{Locational Marginal Pricing for Adaptive Robust Look-Ahead Dispatch with Causal Affine Recourse}

\author{
Aidan Looney,~\IEEEmembership{Graduate Student Member,~IEEE,}
Qian Zhang,~\IEEEmembership{Graduate Student Member,~IEEE,}
and\\
Le Xie,~\IEEEmembership{Fellow,~IEEE}%
\thanks{A. Looney, Q. Zhang, and L. Xie are with the John A. Paulson School of
Engineering and Applied Sciences, Harvard University, Allston, MA, USA. Corresponding author: Aidan Looney (aidanlooney@g.harvard.edu).}
}

\begin{document}

\maketitle

\begin{abstract}
This paper develops a marginal pricing mechanism for adaptive robust look-ahead economic dispatch (LAED) under net-load uncertainty. In current market practice, deterministic multi-interval dispatch can misprice flexibility when forecast error is large. Fully adaptive robust (FAR) dispatch captures this uncertainty, but competing worst-case trajectories can imply different marginal values of demand, leaving no single price for market settlement. We propose causal affine recourse (CAR) as a tractable, price-forming approximation. CAR replaces independently optimized trajectory-specific schedules with a causal affine response policy, yielding a market-clearing model that retains the standard energy and congestion decomposition of DC Locational Marginal Pricing. We then demonstrate that CAR-LMP together with a ramp-adjusted current settlement supports dispatch-following and eliminates current-period lost opportunity costs. We solve the CAR problem by deriving a robust counterpart and present computational evidence on exact small DC instances, a ten-generator system, and an IEEE 300-bus public-network case. These simulations show that the proposed prices collapse exactly to deterministic LAED-LMP when robustness is inactive, remain stable in loose-ramp regimes, and produce economically meaningful price shifts when flexibility is scarce. The proposed pricing mechanism provides a practical bridge between adaptive robust dispatch and market-based pricing.
\end{abstract}


\section{Introduction}

High renewable penetration and increasingly uncertain demand increase net-load volatility and the operational value of ramping capability. A dispatch that is inexpensive in the current interval can create costly or infeasible ramping obligations later, motivating Independent System Operators (ISOs) to adopt rolling-window look-ahead dispatch formulations that co-optimize current and future intervals \cite{xie2008,price2011,hua2019,zhao2020}. Although look-ahead economic dispatch (LAED) internalizes these intertemporal constraints, the traditional Locational Marginal Price (LMP) can fail to support dispatch-following incentives without adjusted settlements \cite{guo2021}. This issue has motivated multi-interval pricing methods \cite{thatte2014,hua2019,zhao2020,guo2021,Chen2021,hogan2020} and explicit ramp or flexibility products in markets such as MISO and CAISO \cite{navid2013,ela2016,wang2017review,wu2016,zhang2026_PSCC}.

Most LAED pricing analyses remain tied to deterministic forecasts even though real-time dispatch faces renewable-generation and demand uncertainty. Stochastic and risk-sensitive models use scenario distributions or tail-risk measures \cite{conejo2010,gu2017,madavan2026_risk_robust_sced}, whereas robust optimization protects against a prescribed uncertainty set \cite{bertsimas2013scuc,lorca2016uc,lorca2015,bertsimas2025pricing}. Standard robust models, referred to here as non-adaptive robust (NAR) models, can be overly conservative because their future decisions do not adjust to realized uncertainty. Adaptive robust optimization (ARO), by contrast, allows future recourse decisions after the current dispatch is fixed and the uncertainty has been observed \cite{bertsimas2022robust}.

ARO has been applied to several power-system applications including two-stage unit commitment (UC), reserve coordination, and LAED. Bertsimas \emph{et al.}\ proposed adaptive robust UC in \cite{bertsimas2013scuc} and Bertsimas and Koulouras \cite{bertsimas2025pricing} address pricing in this problem by deriving LMPs and uplift rules for adaptive robust UC under linear decision rules. Related robust UC work develops uncertainty marginal prices and transmission-reserve interpretations \cite{ye2017_ump}. Warrington \emph{et al.}\ use linear decision rules to co-optimize nominal schedules and affine reserve responses and price the resulting policy-based reserve products \cite{warrington2013_ldr_reserves}. Lorca and Sun \cite{lorca2015} propose adaptive robust LAED with dynamic uncertainty sets, but do not address its marginal-demand pricing or market settlement.

We take the two-stage adaptive robust LAED \cite{lorca2015}, denoted FAR, as the theoretical optimality reference for cost and dispatch. When simultaneously active worst-case demand trajectories have different current-dispatch sensitivities, its current-demand marginal can be set-valued. CAR instead restricts the recourse to a causal affine policy as the price-forming market model and, under DC constraints, retains the standard energy and congestion LMP decomposition. Unlike adaptive robust UC pricing \cite{bertsimas2025pricing} and policy-based reserve pricing \cite{warrington2013_ldr_reserves}, we address marginal energy pricing and current-period settlement for adaptive robust LAED.

The contributions of this paper are threefold.
\begin{enumerate}
    \item We characterize the current-period marginal-demand subdifferential of the fully adaptive formulation in~\cite{lorca2015} and show how active worst-case trajectories with different current-dispatch sensitivities can produce a non-singleton price set.
    \item We formulate CAR as a tractable price-forming model and show that its shared-policy structure prevents competing worst-case realizations from producing different marginal values of demand. Subject to the ordinary dual-uniqueness qualification of deterministic market clearing, CAR yields a single, scenario-independent DC LMP with the standard energy and congestion decomposition. We also derive a ramp-adjusted current settlement that supports dispatch-following incentives and eliminates current-period lost opportunity cost.
    \item We demonstrate computationally that CAR collapses to deterministic LMP when robustness is inactive, produces scarcity-driven price and lost-opportunity-cost effects under tight ramping, and avoids NAR conservatism while retaining a measurable gap to FAR.
\end{enumerate}

\section{Dispatch Models}
\label{sec:dispatch_models}

In dispatch pricing problems, we separate the current interval from the remaining look-ahead window, consistent with the two-stage master/recourse structure used in ARO and exemplified by \cite{lorca2015}. Let $x$ denote the generation dispatch implemented at the current time $t$. A rolling look-ahead solve optimizes the current interval together with $H$ future intervals. Define the time windows, indexed by $\tau$, as
\[
    \mathcal W := \{0,\ldots,H\},
    \qquad
    \mathcal W^+ := \{1,\ldots,H\}.
\]
For each interval $t+\tau$, let $\hat d_{t+\tau}\in\mathbb R^{|\mathcal N|}$ denote the nominal nodal net-load vector. We define the full look-ahead nominal nodal net-load trajectory and the future net-load by
\begin{align*}
&\hat d
:=
\left(
\hat d_t,\hat d_{t+1},\ldots,\hat d_{t+H}
\right), \qquad \hat d^{+}
:=
\left(
\hat d_{t+1},\ldots,\hat d_{t+H}
\right).
\end{align*}

\subsection{Nominal Look-Ahead Economic Dispatch}

The nominal networked LAED jointly selects the generation vector $x_{t+\tau}:=(x_{g,t+\tau})_{g\in\mathcal G}$ and nodal load-shortage vector
$s_{t+\tau}:=(s_{n,t+\tau})_{n\in\mathcal N}$ over the look-ahead horizon:
\begin{equation}
\begin{aligned}
V_{\mathrm N}(\hat d):={}&
\min_{\substack{x_{t+\tau},\,s_{t+\tau}\\
                \tau\in\mathcal W}}
\sum_{\tau\in\mathcal W}
\left(
    c^\top x_{t+\tau}
    +C^{\rm shed}\mathbf 1^\top s_{t+\tau}
\right)
\\
&\text{s.t.}\quad \forall\tau\in\mathcal W:
\\[-0.2em]
&\quad
\mathbf 1^\top
\left(
    Gx_{t+\tau}+s_{t+\tau}-\hat d_{t+\tau}
\right)=0,
\\
&\quad
-\overline f
\leq
\Pi\left(
    Gx_{t+\tau}+s_{t+\tau}-\hat d_{t+\tau}
\right)
\leq\overline f,
\\
&\quad
\underline x
\leq x_{t+\tau}
\leq\overline x,
\qquad
s_{t+\tau}\geq0,
\\
&\quad
-\underline r
\leq x_{t+\tau}-x_{t+\tau-1}
\leq\overline r.
\end{aligned}
\label{eq:nominal-laed}
\end{equation}
Here $c$ is the vector of marginal generation costs, $C^{\rm shed}$ is the value of lost load, $G$ is the generator-to-bus incidence matrix, $\Pi$ is the Power Transfer Distribution Factor (PTDF) matrix, and $\hat d_{t+\tau}$ is the nominal nodal net-load forecast. All capacity, ramping, and line-flow inequalities are applied componentwise. The interval $\tau=0$ is implemented, and its ramp constraints couple $x_t$ to the previously implemented dispatch $x_{t-1}$, which is treated as a parameter.

\subsection{Fully Adaptive Robust Look-Ahead Dispatch}

The nominal LAED formulation optimizes against a single forecast trajectory. The adaptive robust formulation instead fixes only the dispatch implemented in the current interval and evaluates future operating cost against an uncertainty set of admissible future net-load trajectories. Thus the current dispatch is the first-stage decision, while future dispatch and shortage decisions are second-stage recourse actions.

For the robust and pricing derivations, let $x$ abbreviate the implemented current dispatch vector $x_t$, let $s_t$ denote the current shortage, and let $y$ collect all future dispatch, shortage, and auxiliary DC-network variables over $\mathcal W^+$. In an equivalent compact nodal-balance representation, the future part of the nominal LAED can be written as
\[
    Ay=\hat d^{+},\qquad Fy\le f+Ex.
\]
Here $A$ collects the future nodal balance equations, while $Fy\le f+Ex$ collects future capacity, ramping, line-flow, shortage, and other operating constraints. The matrix $E$ denotes the dependence of future inequality right-hand sides on the implemented current dispatch.

Future uncertainty is modeled as a deviation from the nominal future net-load trajectory. Let $\xi\in\Xi\subset\mathbb R^m$ denote the primitive uncertainty vector, where $\Xi$ is compact, convex, and contains the origin. The matrix
$M\in\mathbb R^{H|\mathcal N|\times m}$ maps primitive uncertainty into stacked nodal future net-load deviations, so the realized future trajectory is
\[
    \hat d^{+}+M\xi.
\]
The theoretical results below apply to any fixed compact convex uncertainty set $\Xi$. For marginal pricing, $\Xi$ and its loading matrix $M$ are calibrated before market clearing and held fixed when nominal demand is perturbed.

For a fixed current dispatch $x$ and uncertainty realization $\xi$, the future recourse problem is
\begin{equation}
    q(x,\xi;\hat d^+)
    :=
    \begin{aligned}[t]
        \min_y \quad & b^\top y \\
        \text{s.t.}\quad
        \lambda:\quad
        & Ay = \hat d^+ + M\xi, \\
        \mu:\quad
        & Fy \le f + Ex .
    \end{aligned}
    \label{eq:recourse-lp}
\end{equation}
Here $b^\top y$ is the future operating cost, and $\lambda$ and $\mu$ denote the dual multipliers associated with the future balance equations and operating inequalities, respectively.

The fully adaptive robust continuation value is
\begin{equation}
    Q_{\mathrm{FAR}}(x;\hat d^{+})
    :=
    \max_{\xi\in\Xi}
    q(x,\xi;\hat d^{+}).
    \label{eq:Q-fa}
\end{equation}
The corresponding fully adaptive robust LAED problem is
\begin{equation}
\begin{aligned}
V_{\mathrm{FAR}}(\hat d_t,\hat d^+)
:=
\min_{(x,s_t)\in\mathcal X_t(\hat d_t)}\quad
&c^\top x+C^{\rm shed}\mathbf 1^\top s_t \\
&+Q_{\mathrm{FAR}}(x;\hat d^{+}),
\end{aligned}
    \label{eq:full-adaptive}
\end{equation}
where $\mathcal X_t(\hat d_t)$ denotes the feasible set for the implemented current interval in variables $(x,s_t)$, including current capacity, backward-ramping, system-balance, network, and shortage-nonnegativity constraints. Its dependence on current nominal demand is shown explicitly because the current-period marginal price is obtained by perturbing $\hat d_t$ while holding $\hat d^+$ and the calibrated uncertainty model fixed.

When $\Xi$ is a compact polyhedron with extreme points $\{\xi^{(1)},\ldots,\xi^{(K)}\}$, the fully adaptive continuation value admits the exact master representation \cite[Th.~2.65]{ruszczynski2006_nonlinear}:
\begin{equation}
\label{eq:far-master-extreme}
\begin{alignedat}{3} Q_{\mathrm{FAR}}(x;\hat d^{+}) &= \min_{\eta,\{y^{(k)}\}_{k=1}^K} \eta & \quad \\ \text{s.t.}\quad & \eta \ge b^\top y^{(k)}, &\qquad& k=1,\ldots,K, \\ & Ay^{(k)} = \hat d^{+}+M\xi^{(k)}, && k=1,\ldots,K, \\ & Fy^{(k)} \le f+Ex, && k=1,\ldots,K. \end{alignedat}
\end{equation}
Each extreme realization receives its own recourse copy $y^{(k)}$, and the epigraph variable $\eta$ captures the worst-case future operating cost. In practice, the relevant extreme realizations can be generated iteratively through a separation problem, as in \cite{lorca2015}.

Each $y^{(k)}$ is selected for a complete future uncertainty trajectory $\xi^{(k)}$. Thus, \eqref{eq:far-master-extreme} is a two-stage fully adaptive optimality reference where only the current dispatch is implemented, and the problem is re-solved in the next rolling window.

For the computational studies, $\Xi$ is a fixed nodal dynamic-budget forecast-error set over uncertain buses $\mathcal N_u$. For each uncertain bus $n$ and future interval $\tau=1{:}H$, $\xi_{n,t+\tau}$ denotes the net-load deviation. The uncertainty vector stacks these deviations across buses and time, and $M$ embeds them into the full stacked nodal net-load vector. Auxiliary variables $u_{n,t+\tau}$ bound innovations in the forecast-error process but are not recourse-policy variables. The dynamic budgeted uncertainty set is given as follows:
\begin{subequations}
\label{eq:dynamic-uncertainty-set}
\begin{align}
|\xi_{n,t+1}|
&\le u_{n,t+1},
&& n\in\mathcal N_u,
\label{eq:dyn-init}\\
|\xi_{n,t+\tau}-\rho\xi_{n,t+\tau-1}|
&\le u_{n,t+\tau},
&& n\in\mathcal N_u,\ \tau=2{:}H,
\label{eq:dyn-persistence}\\
0\le u_{n,t+\tau}
&\le \bar\sigma_{n,t+\tau},
&& n\in\mathcal N_u,\ \tau=1{:}H,
\label{eq:dyn-bounds}\\
\sum_{\tau=1}^{H}
\sum_{n\in\mathcal N_u}
\frac{u_{n,t+\tau}}{\bar\sigma_{n,t+\tau}}
&\le \Gamma,
\label{eq:dyn-budget}\\
\hat d_{n,t+\tau}+\xi_{n,t+\tau}
&\ge 0,
&& n\in\mathcal N_u,\ \tau=1{:}H.
\label{eq:dyn-lower}
\end{align}
\end{subequations}
Here $\rho\in[0,1)$ controls temporal persistence, $\Gamma$ is the total normalized innovation budget, and
\[
    \bar\sigma_{n,t+\tau}
    =
    \sigma_{\mathrm{rel}}
    \max\{|\hat d_{n,t+\tau}|,1\}\sqrt{\tau},
\]
measures the nodal innovation budget. Constraint \eqref{eq:dyn-lower} prevents the realized nodal net load from being driven below zero. This construction adapts the dynamic uncertainty-set architecture of \cite{lorca2015}, but instead of modeling wind and load forecast errors separately as they do, we model nodal net-load deviations. The dynamic-budget studies use the persistence parameter $\rho=0.8$, with all calibration parameters fixed before market clearing.

\subsection{Causal Affine Recourse Approximation}

We now introduce causal affine recourse as a price-forming approximation of the FAR formulation. CAR replaces the independently optimized trajectory-specific recourse actions with one affine policy and imposes block-lower-triangular causality. Partition $\xi$ into time blocks $\xi=(\xi_{t+1},\dots,\xi_{t+H})$ and restrict the response matrix $Y$ to block-lower-triangular form. The dispatch and shortage decisions at time $t+\tau$ then depend only on uncertainty revealed through that interval, enforcing non-anticipativity. The causal affine policy is
\begin{equation}
    y(\xi)=y^0+Y\xi.
    \label{eq:affine-policy}
\end{equation}
The resulting affine robust problem is
\begin{align}
V_{\mathrm{Aff}}(\hat d):={}&
\min_{\substack{(x,s_t)\in\mathcal X_t(\hat d_t)\\
                y^0,Y,\eta}}
\left\{
c_t^\top x
+C^{\rm shed}\mathbf 1^\top s_t
+\eta
\right\}
\label{eq:affine-problem}
\\
&\text{s.t.}\quad
A(y^0+Y\xi)
=\hat d^+ +M\xi,
\quad \forall\xi\in\Xi,
\label{eq:affine-feas-1}
\\
&\phantom{\text{s.t.}\quad}
F(y^0+Y\xi)
\leq f+Ex,
\quad \forall\xi\in\Xi,
\label{eq:affine-feas-2}
\\
&\phantom{\text{s.t.}\quad}
\eta
\geq b^\top(y^0+Y\xi),
\quad \forall\xi\in\Xi,
\label{eq:affine-feas-3}
\\
&\phantom{\text{s.t.}\quad}
Y\in\mathcal Y_{\mathrm{causal}}.
\label{eq:affine-feas-4}
\end{align}
Constraints \eqref{eq:affine-feas-1}--\eqref{eq:affine-feas-3} enforce, respectively, power balance, operating feasibility, and the worst-case affine-cost epigraph for every $\xi\in\Xi$. The operating constraints in
\eqref{eq:affine-feas-2} include capacity, ramping, network, and shortage limits, while $\mathcal Y_{\mathrm{causal}}$ contains the block lower triangular response matrices.

Analogously to \eqref{eq:Q-fa}, define the fixed-$x$ affine continuation value as
\begin{equation}
    Q_{\mathrm{Aff}}(x;\hat d^+)
    :=
    \min_{y^0,Y,\eta}
    \left\{
        \eta :
        \eqref{eq:affine-feas-1}\text{--}\eqref{eq:affine-feas-4}
    \right\},
    \label{eq:q-aff}
\end{equation}

The two-stage FAR formulation provides the dispatch optimality reference, while CAR is the causal market-clearing formulation used to price the marginal value of demand. Its support-function robust counterpart \cite{bertsimas2022robust} exposes the nominal-balance and network-constraint multipliers directly. Specifically, each robust affine inequality $a_0+a^\top\xi\le0$ is replaced by $a_0+\sigma_\Xi(a)\le0$, where $\sigma_\Xi(a):=\max_{\xi\in\Xi}a^\top\xi$, and robust equalities are enforced by coefficient matching. Under polyhedral uncertainty of the form $\Xi=\{\xi:H_\Xi\xi\le h_\Xi\}$, the robust counterpart becomes
\[
\sigma_\Xi(a)=\min_{\nu\ge0}\{h_\Xi^\top\nu:H_\Xi^\top\nu=a\}.
\]
This reformulation converts \eqref{eq:affine-problem} into the finite-dimensional linear program used in the computational studies.

Affine decision rules are exact under some special structures \cite{marandi2018}, but not in general. Even when uncertainty appears only in the right-hand sides, affine policies can have an arbitrarily large optimality gap \cite{xu2018}. Networked LAED contains uncertain balance equations and uncertainty-dependent ramping and transmission constraints. Section~\ref{sec:computational_simulations} therefore compares CAR with FAR on selected cases.

\section{Set-Valued Marginal Pricing Under FAR}
\label{sec:fa-set-valued-prices}

For current-period settlement, the relevant value is $V_{\mathrm{FAR}}(\hat d_t,\hat d^+)$ as a function of current nominal nodal demand, holding $\hat d^+$, $\Xi$, $M$, bids, and network data fixed. Active worst-case trajectories affect this value through their sensitivities to the implemented dispatch $x$. Binding intertemporal constraints can transmit different sensitivities to the current-stage Karush--Kuhn--Tucker (KKT) conditions, producing multiple current-period marginal prices.

Let $\Xi$ be a compact polyhedron with extreme points $\{\xi^{(1)},\ldots,\xi^{(K)}\}$. For each extreme point, define the scenario recourse value
\begin{equation}
\begin{aligned}
\phi_k(x;\hat d^+)
:=
\min_y\quad &b^\top y\\
\text{s.t.}\quad
&Ay=\hat d^+ +M\xi^{(k)},\\
&Fy\leq f+Ex,
\end{aligned}
\qquad k=1,\ldots,K.
\label{eq:scenario-value}
\end{equation}
Since the recourse value is convex in $\xi$, its maximum over $\Xi$ is attained at an extreme point \cite[Th.~2.65]{ruszczynski2006_nonlinear}. Therefore,

\begin{equation}
Q_{\mathrm{FAR}}(x;\hat d^+)
=
\max_{k=1,\ldots,K}
\phi_k(x;\hat d^+).
\label{eq:qfa-max-scenarios}
\end{equation}
Where finite in a neighborhood, each $\phi_k$ is convex and polyhedral, hence piecewise affine, as a function of $(x,\hat d^+)$, while $Q_{\mathrm{FAR}}$ is their pointwise maximum \cite[Prop.~3.55(b)]{rockafellar1998}. Active worst-case scenarios can therefore have different sensitivities to the implemented dispatch. The following theorem relates these sensitivities to the set of current-period marginal prices.

\begin{theorem}[Current-period marginal prices under FAR]
\label{thm:far-current-price}
Consider the FAR problem \eqref{eq:full-adaptive} with the extreme-point representation \eqref{eq:far-master-extreme}. Suppose that the finite master problem and its dual attain a common finite optimal value, that $V_{\mathrm{FAR}}(\cdot,\hat d^+)$ is finite in a neighborhood of $\hat d_t$, and that every $\phi_k(\cdot;\hat d^+)$ is finite in a neighborhood of every optimal current dispatch.

Let $\mathcal Z_{\mathrm{FAR}}^\star(\hat d_t,\hat d^+)$ denote the set of dual-optimal multiplier vectors of the finite master problem. For $\zeta\in\mathcal Z_{\mathrm{FAR}}^\star$, let $\omega_t(\zeta)$ be the multiplier on the current balance constraint $\mathbf 1^\top\hat d_t-\mathbf 1^\top(Gx+s_t)=0$, and let $\mu_t^+(\zeta),\mu_t^-(\zeta)\geq0$ be the multipliers on the current upper and lower line-flow constraints. Define
\[
p_t(\zeta)
:=
\omega_t(\zeta)\mathbf 1
-\Pi^\top
\left(
\mu_t^+(\zeta)-\mu_t^-(\zeta)
\right).
\]
Then the set of current-period marginal prices is
\begin{equation}
\partial_{\hat d_t}
V_{\mathrm{FAR}}(\hat d_t,\hat d^+)
=
\left\{
p_t(\zeta):
\zeta\in
\mathcal Z_{\mathrm{FAR}}^\star
(\hat d_t,\hat d^+)
\right\}.
\label{eq:far-current-subdifferential}
\end{equation}
Consequently, FAR yields a unique current-period LMP if and only if $p_t(\zeta)$ is constant over $\mathcal Z_{\mathrm{FAR}}^\star$.

For any optimal current dispatch $x^\star$, define the active worst-case set

\[
\mathcal A(x^\star,\hat d^+)
:=
\arg\max_{k=1,\ldots,K}
\phi_k(x^\star;\hat d^+).
\]
For each scenario $k$,
\begin{equation}
\partial_x\phi_k(x^\star;\hat d^+)
=
\left\{
E^\top\mu:
(\lambda,\mu)
\text{ is dual optimal}
\right\}.
\label{eq:scenario-current-sensitivity}
\end{equation}

Here dual optimality refers to the dual of \eqref{eq:scenario-value}, evaluated at $(x^\star,\hat d^+)$ for scenario $k$. The FAR recourse subdifferential is therefore
\begin{equation}
\partial_xQ_{\mathrm{FAR}}(x^\star;\hat d^+)
=
\operatorname{conv}
\left(
\bigcup_{k\in\mathcal A(x^\star,\hat d^+)}
\partial_x\phi_k(x^\star;\hat d^+)
\right).
\label{eq:far-recourse-subdifferential}
\end{equation}
Equivalently, any recourse subgradient $g^\star\in \partial_xQ_{\mathrm{FAR}}(x^\star;\hat d^+)$ appearing in the current-dispatch KKT conditions can be written as
\begin{equation}
\begin{aligned}
g^\star
&=
\sum_{k\in\mathcal A(x^\star,\hat d^+)}
\theta_k g_k,\\
\sum_{k\in\mathcal A(x^\star,\hat d^+)}
\theta_k
&=1,
\qquad
\theta_k\geq0,
\qquad
g_k\in
\partial_x\phi_k(x^\star;\hat d^+).
\end{aligned}
\label{eq:far-sensitivity-combination}
\end{equation}
The weights $\theta_k$ can be chosen as the optimal multipliers on the scenario epigraph constraints in \eqref{eq:far-master-extreme}.
\end{theorem}

\begin{proof}
Let $\mathcal Z_{\mathrm{FAR}}$ denote the full dual-feasible set of the finite master problem. The terms in its dual objective that contain current nominal demand are $p_t(\zeta)^\top\hat d_t$; collect all remaining terms in $a(\zeta;\hat d^+)$. Strong duality gives
\[
V_{\mathrm{FAR}}(\hat d_t,\hat d^+)
=
\sup_{\zeta\in\mathcal Z_{\mathrm{FAR}}}
\left\{
a(\zeta;\hat d^+)
+p_t(\zeta)^\top\hat d_t
\right\}.
\]

Thus $V_{\mathrm{FAR}}(\cdot,\hat d^+)$ is the pointwise supremum of affine functions of current demand. The convex optimal-value sensitivity formula \cite[Th.~4.26]{ruszczynski2006_nonlinear} identifies its subdifferential with the price vectors induced by dual-optimal solutions. Because $\mathcal Z_{\mathrm{FAR}}^\star$ is a polyhedron and $p_t(\cdot)$ is linear, its image is also a polyhedron and is therefore closed and convex. This proves \eqref{eq:far-current-subdifferential} and the uniqueness statement.

For each scenario, strong duality gives
\begin{equation}
\begin{aligned}
\phi_k(x;\hat d^+)
=
\max_{\lambda,\mu}\quad
&\lambda^\top(\hat d^+ +M\xi^{(k)})
+\mu^\top(f+Ex)\\
\text{s.t.}\quad
&A^\top\lambda+F^\top\mu=b,
\qquad \mu\leq0.
\end{aligned}
\label{eq:scenario-dual}
\end{equation}

Applying the same optimal-value sensitivity formula yields \eqref{eq:scenario-current-sensitivity}. The finite-maximum subdifferential rule applied to \eqref{eq:qfa-max-scenarios} \cite[Th.~2.87]{ruszczynski2006_nonlinear} then gives \eqref{eq:far-recourse-subdifferential}, and the convex-hull representation gives \eqref{eq:far-sensitivity-combination}.

The same weights arise from the finite master problem. If $\theta_k$ is the multiplier on the $k$-th scenario epigraph constraint, stationarity with respect to $\eta$ gives $\sum_k\theta_k=1$, while complementary slackness gives $\theta_k=0$ for every $k\notin\mathcal A(x^\star,\hat d^+)$. For $\theta_k>0$, normalizing the corresponding scenario multipliers by $\theta_k$ produces a dual-optimal solution of \eqref{eq:scenario-dual} and hence a subgradient $g_k\in\partial_x\phi_k(x^\star;\hat d^+)$. Scenarios with $\theta_k=0$ make no contribution to $g^\star$.
\end{proof}

Multiple active worst-case scenarios do not by themselves imply multiple current prices. Their current-dispatch sensitivities must differ, and binding current-to-future constraints must transmit those differences into distinct full KKT price vectors. Thus, if two dual-optimal solutions $\zeta^{(1)}$ and $\zeta^{(2)}$ satisfy $p_t(\zeta^{(1)})\neq p_t(\zeta^{(2)})$, FAR does not select a unique current-period price. Conversely, the price remains unique when $p_t(\zeta)$ is constant over the dual-optimal set, even if other multipliers are nonunique.

Operationally, this distinction matters because the market operator must publish a single, reproducible price for the implemented interval. When price multiplicity arises from competing worst-case trajectories, different dual-optimal weightings of those trajectories can support different settlement prices for the same cleared dispatch and objective value. This is a price-formation issue rather than a failure of dispatch feasibility or cost optimality: the optimization does not determine the settlement price without an additional exogenous selection rule. Such a rule would affect generator payments and the compensation attributed to ramping flexibility while weakening the direct connection between the market-clearing model and the published LMP. This motivates a price-forming robust formulation that retains uncertainty-aware dispatch while producing a single current-period DC LMP.

\section{Marginal Pricing for Causal Affine Recourse}
\label{sec:car_marginal_pricing}

The causal affine restriction replaces the infinite-dimensional recourse-policy space with the finite-dimensional variables $(y^0,Y)$ through $y(\xi):=y^0+Y\xi$. Assume that $\Xi\subseteq\mathbb R^m$ is full dimensional, as in the computational studies. Because two affine functions that agree on a full-dimensional set have identical coefficients, the robust balance equation \eqref{eq:affine-feas-1} is equivalent to
\begin{align}
    Ay^0 &= \hat d^+,
    \label{eq:nominal-balance}\\
    AY &= M.
    \label{eq:response-balance}
\end{align}

Equation~\eqref{eq:nominal-balance} associates nominal demand with the affine-policy intercept, while \eqref{eq:response-balance} enforces consistency of the response coefficients. In the DC specialization, nominal demand also enters the nominal components of the line-flow constraints. The corresponding balance and line-flow multipliers therefore provide the energy and congestion components of the CAR marginal price.

\subsection{DC-Networked Affine Recourse and Adaptive Robust LMP}

Let $G$ be the generator-to-bus incidence matrix, $\Pi$ the PTDF matrix, and let nodal uncertainty at interval $t+\tau$ enter as $\hat d_{t+\tau}+D_{t+\tau}\xi$, where $D_{t+\tau}$ is the time-$\tau$ block row of $M$. Generation and shortage use the affine policies
\[
    x_{t+\tau}(\xi) = x_{t+\tau}^0+R_{t+\tau}\xi, \qquad s_{t+\tau}(\xi) = s_{t+\tau}^0+S_{t+\tau}\xi,
\]
with $D_t=R_t=S_t=0$ because current net load is known. Block-lower-triangular $(R,S)$ impose causality. After coefficient matching and support-function reformulation, the robust counterpart enforces balance, line-flow, capacity, shortage, ramping, causality, and the worst-case cost epigraph.

For each interval, robust balance is equivalent to
\begin{align}
    \mathbf 1^\top
    \left(
        Gx^0_{t+\tau} + s^0_{t+\tau} - \hat d_{t+\tau}
    \right) &=0,
    \label{eq:dc-nom-energy}\\
    \mathbf 1^\top
    \left(
        GR_{t+\tau} + S_{t+\tau} - D_{t+\tau}
    \right) &=0.
    \label{eq:dc-response-energy}
\end{align}
The first identity governs the nominal dispatch, while the second ensures that the affine response balances every uncertainty realization.

Let $\Delta_{t+\tau}:= GR_{t+\tau}+S_{t+\tau}-D_{t+\tau}$ and apply $\sigma_\Xi$ row-wise to matrices. The robust upper and lower line-flow constraints are
\begin{align}
    &\Pi
    \left(
        Gx_{t+\tau}^0+s^0_{t+\tau}-\hat d_{t+\tau}
    \right) + \sigma_\Xi^{\mathcal L}
    \left(
        \Pi\Delta_{t+\tau}
    \right) \leq \overline f_{t+\tau},
    \label{eq:dc-flow-upper}\\
    &-\Pi
    \left(
        Gx_{t+\tau}^0+s^0_{t+\tau}-\hat d_{t+\tau}
    \right) + \sigma_\Xi^{\mathcal L}
    \left(
        -\Pi\Delta_{t+\tau}
    \right) \leq \overline f_{t+\tau}.
    \label{eq:dc-flow-lower}
\end{align}
For each $g$ and $\tau\in\mathcal W^+$, the robust capacity and ramp constraints are
\begin{align}
    &x^0_{g,t+\tau} + \sigma_\Xi(R_{g,t+\tau}^\top) \leq \overline x_g,
    \label{eq:dc-capacity-counterpart-up}\\
    &\underline x_g + \sigma_\Xi(-R_{g,t+\tau}^\top) \leq x^0_{g,t+\tau},
    \label{eq:dc-capacity-counterpart-down}\\
    &x^0_{g,t+\tau}-x^0_{g,t+\tau-1} + \sigma_\Xi
    \left(
        (R_{g,t+\tau}-R_{g,t+\tau-1})^\top
    \right) \leq \overline r_g,
    \label{eq:dc-ramp-counterpart-up}\\
    &x^0_{g,t+\tau-1} - x^0_{g,t+\tau} + \sigma_\Xi
    \left(
        (R_{g,t+\tau-1}-R_{g,t+\tau})^\top
    \right) \leq \underline r_g.
    \label{eq:dc-ramp-counterpart-down}
\end{align}
Here $R_{g,t}=0$, so the $\tau=1$ constraints couple every admissible first-future response to the implemented current dispatch. Robust shortage nonnegativity is

\[
    \sigma_\Xi(-S_{n,t+\tau}^\top) \leq s^0_{n,t+\tau}.
\]

For the linear costs used below, the future-cost epigraph is
\begin{align}
    &\sum_{\tau\in\mathcal W^+}
    \left(
        c^\top x^0_{t+\tau} + C^{\rm shed}\mathbf 1^\top s^0_{t+\tau}
    \right) \notag\\
    &\quad+\sigma_\Xi
    \left(
        \sum_{\tau\in\mathcal W^+}
        \left(
            R_{t+\tau}^\top c + C^{\rm shed}S_{t+\tau}^\top\mathbf 1
        \right)
    \right) \leq \eta.
    \label{eq:dc-cost-counterpart}
\end{align}
Nominal demand enters directly only through \eqref{eq:dc-nom-energy}--\eqref{eq:dc-flow-lower}; the support-function terms enforce feasibility of the shared affine response.

\begin{theorem}[CAR locational marginal prices and scenario-invariant nominal gradients]
\label{thm:dc-lmp}

Consider the finite-dimensional DC robust counterpart of the causal affine robust dispatch problem with a compact polyhedral uncertainty set $\Xi$. Hold $\Xi$, the uncertainty-loading matrices $D_{t+\tau}$, and all remaining market and network data fixed when perturbing the nominal demand vectors $\hat d_{t+\tau}$. Suppose that the robust counterpart and its dual attain a common finite optimal value and that $V_{\mathrm{Aff}}$ is finite in a neighborhood of $\hat d$.

Fix any $\tau\in\mathcal W$. Let $\mathcal Z^\star(\hat d)$ denote the set of full dual-optimal multiplier vectors. For $\zeta\in\mathcal Z^\star(\hat d)$, let $\omega_{t+\tau}(\zeta)$ be the multiplier on the nominal balance constraint written as
\begin{equation}
    \mathbf 1^\top\hat d_{t+\tau} - \mathbf 1^\top Gx_{t+\tau}^0 - \mathbf 1^\top s_{t+\tau}^0 = 0,
    \label{eq:dc-nom-balance-priced}
\end{equation}
and let $\mu_{t+\tau}^+(\zeta),\mu_{t+\tau}^-(\zeta) \in\mathbb R_+^{|\mathcal L|}$ be the multipliers on \eqref{eq:dc-flow-upper} and \eqref{eq:dc-flow-lower}. Define
\begin{equation}
    p_{t+\tau}(\zeta) := \omega_{t+\tau}(\zeta)\mathbf 1 - \Pi^\top
    \left(
        \mu_{t+\tau}^+(\zeta) - \mu_{t+\tau}^-(\zeta)
    \right).
    \label{eq:car-dc-price-map}
\end{equation}
Then
\begin{equation}
    \partial_{\hat d_{t+\tau}} V_{\mathrm{Aff}}(\hat d) = 
    \left\{
        p_{t+\tau}(\zeta): \zeta\in\mathcal Z^\star(\hat d)
    \right\}.
    \label{eq:car-dc-subdiff}
\end{equation}
Moreover, multiple active uncertainty realizations do not create different direct CAR nominal-policy gradients. In particular, the robust future-cost function satisfies 
\begin{equation}
\begin{aligned}
    C_{\mathrm{Aff}}^+(y^0,Y)
    &:=
    \max_{\xi\in\Xi}
    b^\top(y^0+Y\xi)\\
    &=
    b^\top y^0+\sigma_\Xi(Y^\top b),
\end{aligned}
\label{eq:car-cost-separation}
\end{equation}
and therefore
\begin{equation}
    \partial_{y^0}
    C_{\mathrm{Aff}}^+(y^0,Y)
    =
    \{b\},
    \label{eq:car-nominal-gradient}
\end{equation}
regardless of the number of realizations attaining the worst-case cost. The same separation applies row-wise to the robust affine constraints above: their support-function terms depend on the response coefficients but not on the nominal-policy variables. Active realizations can therefore change response-coefficient subgradients, but they have the same direct gradient with respect to the nominal block.

Consequently, CAR determines a unique DC LMP at interval $t+\tau$ if and only if the price set in \eqref{eq:car-dc-subdiff} is a singleton. In that case,
\begin{equation}
    \lambda_{t+\tau}^{\mathrm{CAR}} = \omega_{t+\tau}^\star\mathbf 1 - \Pi^\top
    \left(
        \mu_{t+\tau}^{+,\star} - \mu_{t+\tau}^{-,\star}
    \right),
    \label{eq:car-dc-lmp}
\end{equation}
where any full dual-optimal multiplier vector may be used. Uniqueness of the individual balance and line-flow multipliers is sufficient, but not necessary, for uniqueness of $\lambda_{t+\tau}^{\mathrm{CAR}}$.
\end{theorem}

\begin{proof}
Fix $\tau\in\mathcal W$ and let $\zeta$ denote a full dual-feasible multiplier vector. Under the sign conventions in \eqref{eq:dc-nom-balance-priced} and \eqref{eq:dc-flow-upper}--\eqref{eq:dc-flow-lower}, the terms in the Lagrangian containing $\hat d_{t+\tau}$ have coefficient
\[
    \omega_{t+\tau}\mathbf 1 - \Pi^\top
    \left(
        \mu_{t+\tau}^+
        -
        \mu_{t+\tau}^-
    \right)
    =p_{t+\tau}(\zeta).
\]
The dual objective can therefore be written as

\begin{equation}
    \Theta(\zeta;\hat d) =  a(\zeta) +
    \sum_{\kappa\in\mathcal W} p_{t+\kappa}(\zeta)^\top \hat d_{t+\kappa},
    \label{eq:dual-affine-demand}
\end{equation}
where $a(\zeta)$ is independent of $\hat d$. Strong duality gives
\[
    V_{\mathrm{Aff}}(\hat d) = \sup_{\zeta\in\mathcal Z} \Theta(\zeta;\hat d),
\]
where $\mathcal Z$ is the full dual-feasible set. The convex optimal-value sensitivity formula \cite[Th.~4.26]{ruszczynski2006_nonlinear} therefore yields
\[
    \partial_{\hat d_{t+\tau}}
    V_{\mathrm{Aff}}(\hat d)
    =
    \operatorname{conv}
    \left\{
        p_{t+\tau}(\zeta):
        \zeta\in\mathcal Z^\star(\hat d)
    \right\}.
\]
Because $\mathcal Z^\star(\hat d)$ is a polyhedron and $p_{t+\tau}$ is linear, its image is also a polyhedron and hence closed and convex. The convex hull may therefore be omitted, proving \eqref{eq:car-dc-subdiff}.

For the worst-case cost,
\[
\begin{aligned}
    \max_{\xi\in\Xi}
    b^\top(y^0+Y\xi)
    &=
    b^\top y^0
    +
    \max_{\xi\in\Xi}
    (Y^\top b)^\top\xi\\
    &=
    b^\top y^0+\sigma_\Xi(Y^\top b).
\end{aligned}
\]

The support-function term is independent of $y^0$, proving \eqref{eq:car-nominal-gradient}. If several realizations $\xi^{(k)}$ attain the maximum, their convex weights can produce different response-coefficient subgradients of the form
\[
    \sum_k\theta_k b(\xi^{(k)})^\top,
    \qquad
    \theta_k\geq0,
    \qquad
    \sum_k\theta_k=1,
\]
but the corresponding intercept gradient remains $\sum_k\theta_k b=b$.

More generally, each robust affine inequality row can be written as
\[
    \psi_j(z^0,Z;\hat d)
    =
    a_j(z^0,\hat d)
    +
    \sigma_\Xi\!\left(q_j(Z)\right),
\]
where $z^0$ collects the uncertainty-independent current and nominal-policy variables, $Z$ collects the response
coefficients, and $a_j$ is affine. Because the support-function term is independent of $(z^0,\hat d)$,
\[
    \partial_{(z^0,\hat d)}
    \psi_j(z^0,Z;\hat d)
    =
    \left\{
        \nabla_{(z^0,\hat d)}
        a_j(z^0,\hat d)
    \right\}.
\]
Thus, active-realization multiplicity can affect response-coefficient stationarity, but it does not introduce different direct nominal gradients. 

Finally, \eqref{eq:car-dc-subdiff} is a singleton exactly when $p_{t+\tau}$ is constant over $\mathcal Z^\star(\hat d)$. In that case, every dual optimum gives the common price in \eqref{eq:car-dc-lmp}. Distinct dual optima may nevertheless produce the same price because only their image under $p_{t+\tau}$ is economically relevant.
\end{proof}

Theorem~\ref{thm:dc-lmp} identifies the structural price-forming distinction between FAR and CAR. Under FAR,
independently optimized active recourse problems can contribute different current-dispatch sensitivities through
their epigraph weights. Under CAR, all uncertainty realizations are evaluated through a shared policy. Their
response-coefficient subgradients may differ, but their direct nominal-policy gradients are the same. CAR therefore removes the FAR mechanism through which competing worst-case realizations can directly imply different
marginal values of demand. It does not eliminate ordinary price-relevant dual nonuniqueness in the common robust counterpart, just as deterministic LAED does not guarantee unique multipliers in every degenerate instance. This scenario-invariance property comes from the shared affine parameterization and the separation of its nominal and response coefficients, rather than from causality of the recourse.

The corresponding generator stationarity condition makes the intertemporal price signal explicit. For a non-endpoint interval $t+\tau$, let $\alpha_{t+\tau}^{\pm}$ be the capacity multipliers and let $\gamma_{t+\tau}^{\pm}$ be the ramp multipliers linking intervals $t+\tau-1$ and $t+\tau$. Stationarity of the nominal generation dispatch gives
\begin{align}
    &\nabla c_{t+\tau}(x_{t+\tau}^{0,\star})
    -G^\top\lambda_{t+\tau}^{\mathrm{CAR},\star}
    +\alpha_{t+\tau}^{+,\star}-\alpha_{t+\tau}^{-,\star}
    \notag\\
    &\quad+\gamma_{t+\tau}^{+,\star}-\gamma_{t+\tau}^{-,\star}
    -\gamma_{t+\tau+1}^{+,\star}+\gamma_{t+\tau+1}^{-,\star}=0.
    \label{eq:dc-generator-stationarity}
\end{align}
Thus the bus price equals marginal production cost adjusted by capacity scarcity and adjacent-interval ramp opportunity costs. This identity motivates the current-period ramp adjustment below.

Setting $\tau=0$ in \eqref{eq:car-dc-subdiff} gives the CAR counterpart to \eqref{eq:far-current-subdifferential}; when this image is a singleton, \eqref{eq:car-dc-lmp} gives the current nodal settlement $\lambda_t^{\mathrm{CAR}}$. Future uncertainty changes its value through robust-policy and intertemporal multipliers but introduces no trajectory index. The next theorem adds the forward-ramp opportunity cost needed for unit-level dispatch-following.

\begin{theorem}[Ramp-adjusted current settlement supports generator profit maximization]
\label{thm:current-incentive}
Consider one rolling-horizon CAR market clearing at current interval $t$, and
fix a primal-dual optimal solution. Let $x_t^{0,\star}$ be the implemented
current dispatch. For generator $g$, let $n(g)$ denote its bus and define its
current local feasible set by
\begin{equation}
\begin{aligned}
\mathcal X_{g,t}(x_{g,t-1}^{\mathrm{impl}})
:=
\{z_g:\;&
\underline x_g \leq z_g \leq \overline x_g,\\
&
-\underline r_g
\leq z_g-x_{g,t-1}^{\mathrm{impl}}
\leq \overline r_g
\}.
\end{aligned}
\label{eq:current-local-feasible-set}
\end{equation}
Assume the current-period generation cost $c_{g,t}(\cdot)$ is convex and differentiable and that, apart from the first-future ramp constraints priced below, the only unit-specific constraints involving $x_{g,t}^{0}$ are the capacity and backward-ramp constraints defining $\mathcal X_{g,t}(x_{g,t-1}^{\mathrm{impl}})$.

Let $\lambda_{n(g),t}^{\mathrm{CAR},\star}$ be the bus-level CAR  LMP at generator $g$'s bus. Let $\gamma_{g,t+1}^{+,\star},\gamma_{g,t+1}^{-,\star}\ge 0$ be the optimal CAR multipliers on the first future ramp constraints linking the current output $x_{g,t}^{0}$ to the first future affine dispatch $x_{g,t+1}(\xi)$, with stationarity contribution $-\gamma_{g,t+1}^{+,\star}+\gamma_{g,t+1}^{-,\star}$ with respect to $x_{g,t}^{0}$. Define the unit-specific ramp-adjusted current price by
\begin{equation}
    \pi_{g,t}^{\mathrm{CAR},\star}
    :=
    \lambda_{n(g),t}^{\mathrm{CAR},\star}
    +
    \gamma_{g,t+1}^{+,\star}
    -
    \gamma_{g,t+1}^{-,\star}.
    \label{eq:current-ramp-adjusted-price}
\end{equation}
Then the ISO instruction $x_{g,t}^{0,\star}$ solves generator $g$'s current profit-maximization problem:
\begin{equation}
    x_{g,t}^{0,\star}
    \in
    \operatorname*{arg\,max}_{z_g\in\mathcal X_{g,t}(x_{g,t-1}^{\mathrm{impl}})}
    \left\{
        \pi_{g,t}^{\mathrm{CAR},\star} z_g
        -
        c_{g,t}(z_g)
    \right\}.
    \label{eq:current-local-profit}
\end{equation}
Equivalently, the same incentive support can be implemented by settling energy at the bus-level CAR LMP $\lambda_{n(g),t}^{\mathrm{CAR},\star}$ and paying the unit-specific forward-ramp adder
\[
    \left(
        \gamma_{g,t+1}^{+,\star}
        -
        \gamma_{g,t+1}^{-,\star}
    \right)z_g .
\]
\end{theorem}

\begin{proof}
Let
$\alpha_{g,t}^{+,\star},\alpha_{g,t}^{-,\star}\ge 0$ denote the optimal
multipliers on the current upper and lower output bounds for generator $g$,
and let
$\gamma_{g,t}^{+,\star},\gamma_{g,t}^{-,\star}\ge 0$ denote the optimal
multipliers on the backward-ramp constraints defining
$\mathcal X_{g,t}(x_{g,t-1}^{\mathrm{impl}})$. The CAR stationarity condition
for $x_{g,t}^0$ is
\begin{equation}
\begin{aligned}
    0
    =
    &c'_{g,t}(x_{g,t}^{0,\star})
    -
    \lambda_{n(g),t}^{\mathrm{CAR},\star}
    +
    \alpha_{g,t}^{+,\star}
    -
    \alpha_{g,t}^{-,\star}
    +
    \gamma_{g,t}^{+,\star}
    -
    \gamma_{g,t}^{-,\star}
    \\
    &-
    \gamma_{g,t+1}^{+,\star}
    +
    \gamma_{g,t+1}^{-,\star}.
\end{aligned}
\label{eq:current-system-stationarity}
\end{equation}
Substituting \eqref{eq:current-ramp-adjusted-price} gives
\begin{equation}
    c'_{g,t}(x_{g,t}^{0,\star})
    -
    \pi_{g,t}^{\mathrm{CAR},\star}
    +
    \alpha_{g,t}^{+,\star}
    -
    \alpha_{g,t}^{-,\star}
    +
    \gamma_{g,t}^{+,\star}
    -
    \gamma_{g,t}^{-,\star}
    =
    0.
\label{eq:current-local-stationarity}
\end{equation}
Together with primal feasibility in $\mathcal X_{g,t}(x_{g,t-1}^{\mathrm{impl}})$, dual feasibility, and complementary slackness for the current capacity and backward-ramp constraints, \eqref{eq:current-local-stationarity} is exactly the KKT system for the
equivalent convex minimization problem:
\[
    \min_{z_g\in\mathcal X_{g,t}(x_{g,t-1}^{\mathrm{impl}})}
    \left\{
        c_{g,t}(z_g)
        -
        \pi_{g,t}^{\mathrm{CAR},\star}z_g
    \right\}.
\]
Equivalently, $x_{g,t}^{0,\star}$ solves the profit-maximization problem
\eqref{eq:current-local-profit}. Since $c_{g,t}(\cdot)$ is convex and
$\mathcal X_{g,t}(x_{g,t-1}^{\mathrm{impl}})$ is polyhedral, these KKT
conditions are sufficient. The side-payment form is algebraically equivalent because
the forward-ramp term can be written either as part of the resource-specific
linear price
$\pi_{g,t}^{\mathrm{CAR},\star}z_g$
or as a linear side payment added to the bus-level LMP settlement.
\end{proof}

The resulting CAR-LAED settlement combines the bus-level LMP from Theorem~\ref{thm:dc-lmp} with the discriminate forward-ramp adder in Theorem~\ref{thm:current-incentive} \cite{guo2021,Chen2021}. The adder restores the stationarity term associated with positioning current output for first-future-interval ramp feasibility; under the stated assumptions, it eliminates current-period lost opportunity cost without changing dispatch.

\subsection{ISO Implementation}

One CAR solve replaces deterministic LAED on the same DC network model. The ISO settles the current interval using \eqref{eq:car-dc-lmp} and, when needed, the forward-ramp side payment from Theorem~\ref{thm:current-incentive}; no scenario-indexed repricing is required. It then passes the implemented dispatch to the next rolling window and re-solves with updated conditions. The following studies quantify when active ramping and congestion make CAR differ from FAR or NAR.

\section{Computational Simulations}
\label{sec:computational_simulations}

Three studies test the CAR pricing mechanism and its operational cost. A single-bus example isolates the FAR current-price issue; a ten-generator rolling-horizon study compares deterministic LAED, NAR, and CAR across uncertainty budgets and ramp regimes, with FAR comparisons for selected $\Gamma=3$ windows; and an IEEE 300-bus case~\cite{pglibopf} tests adaptivity gaps, computation, and the ramp-adjusted settlement. All models are linear programs implemented in Pyomo and solved with Gurobi.

\subsection{Toy Example}

This example highlights the distinction between the price-forming properties of FAR and CAR, providing intuition for the results of Theorems~\ref{thm:far-current-price} and~\ref{thm:dc-lmp}. Consider a one-bus dispatch with three generators and no transmission constraints. The cheap, slow, and peaker generators have marginal costs (1, 4, 8) \$/MW, capacities (13, 14, 20) MW, respectively. Only the slow generator is ramp-limited at 4 MW per interval. There are no binding constraints between the previous and current interval, and load in the current interval is 18 MW. This example expresses marginal costs and LMPs in \$/MW, not \$/MWh because it is not restricted to 1-hour time discretizations. We define how uncertainty enters the net load as the following simplification of the dynamic uncertainty set with uncertainty only in $t+2, t+3$:
\begin{align}
d^+(\xi) = (10, 14+14\xi, 35-7\xi), \quad \xi \in [0,1]
\end{align}

Both formulations choose the obvious current-interval dispatch (13, 5, 0) with cost \$33. The FAR and CAR differ in their future dispatch as provided below:

\emph{FAR}: As a scenario-dependent, perfect-information formulation, FAR allows separate recourse for each worst-case uncertainty realization. Table \ref{tab:toy-dispatch} provides the obvious recourse dispatches for each generator type, where $y_c^+, y_s^+, y_p^+$ denote the recourse dispatch of the cheap, slow, and peaker generators, respectively. The bounds of the uncertainty set yield equally bad worst-case load perturbations as demonstrated by the same optimal cost. The important distinction here is the derivation of the LMP for each scenario. For $\xi = 0$, the LMP is \$6/MW because one more current slow MW propagates through all three future intervals. Under $\xi=1$, it replaces cheap generation at $t+1$ and peaker generation at $t+2$, which is represented additively across time periods as follows:
\[
\lambda_{\xi=0}^{FAR}=4+3+3-4 = \$6, \quad \lambda_{\xi=1}^{FAR}=4+3-4+0=\$3
\]

\begin{table}[t]
\centering
\caption{Toy Example Dispatch $(t, t+1, t+2, t+3)$ Total Optimal Cost (USD) and LMP (USD/MWh).}
\label{tab:toy-dispatch}
\scriptsize
\setlength{\tabcolsep}{2.2pt}
\begin{tabular}{@{}ccccccc@{}}
\hline
Model & $\xi$  & $y_c$ & $y_s$ & $y_p$ & Total Cost & $\lambda_t$\\
\hline
FAR & $0$ &  $(13, 9,9,13)$  & $(5, 1,5,9)$   & $(0, 0,0,13)$ & $\$228$ & $\$6$\\
FAR & $1$ &  $(13, 1,13,13)$ & $(5, 9,13,14)$ & $(0, 0,2,1)$  & $\$228$ & $\$3$\\
CAR & $0$ &  $(13, 5, 5, 13)$  & $(5, 5,9,13)$  & $(0, 0,0,9)$  & $\$236$ & $\$4$\\
CAR & $1$  & $(13, 5,13,13)$ & $(5, 5,9,13)$  & $(0, 0,6,2)$  & $\$236$ & $\$4$\\
\hline
\end{tabular}
\end{table}

\emph{CAR}: CAR instead chooses a common causal output policy: $y_{s,t+1} \in [1,9]$, which is simplified to $y$ for conciseness. We choose to make recourse decisions functions of $y$ because the slow plant expresses the pricing difference through binding intertemporal constraints. Since $d_{t+1}=10$, the $t+1$ dispatch is given by $(10-y, y, 0)$ for cheap, slow, and peaker generators respectively. Let $C_\xi(y)$ be the minimum cost of the three future intervals conditional on the uncertainty realization $\xi$ and common $y$, and the policy changes at $y=6$, where the relevant capacity constraint binds. For $1 \leq y \leq 6$, the slow unit does not reach capacity by $t+3$ and has recourse dispatch $(y, y+4, y+8)$ for both $\xi$. Using the dispatch regions in Table~\ref{tab:toy-value-function}, the two worst-case cost functions $(C_0(y), C_1(y))$ are the summations of cost across generators and time horizon:
\begin{equation}
\bigl(C_0(y),C_1(y)\bigr)
=
\begin{cases}
\bigl(193+2y,\;228-5y\bigr),
    & 1\leq y\leq6,\\
\bigl(187+3y,\;204-y\bigr),
    & 6\leq y\leq9.
\end{cases}
\label{eq:toy-endpoint-costs}
\end{equation}

\begin{table}[t]
\centering
\caption{Toy Example Recourse Dispatch $(t+1, t+2, t+3)$ as a function of $y:=y_{s,t+1}$.}
\label{tab:toy-value-function}
\scriptsize
\setlength{\tabcolsep}{2.2pt}
\begin{tabular}{@{}ccccc@{}}
\hline
 $\xi$  & $y$-bounds& $y_c^+$ & $y_s^+$ & $y_p^+$ \\
\hline
$0$ & $1\leq y\leq 6$ & $(10-y,10-y,13)$ & $(y,y+4,y+8)$ & $(0, 0, 14-y)$ \\
$0$ & $6\leq y\leq 9$ & $(10-y,4,13)$ & $(y,10,14)$ & $(0, 0, 8)$ \\
$1$ & $1\leq y\leq 6$ & $(10-y,13,13)$ & $(y,y+4,y+8)$ & $(0,11-y,7-y)$ \\
$1$ & $6\leq y\leq 9$ & $(10-y,13,13)$ & $(y,y+4,14)$ & $(0, 11-y, 1)$ \\
\hline
\end{tabular}
\end{table}

Because every causal policy must select one common $y$, its worst-case future cost is bounded below by
\begin{align}
    \min_{1\leq y\leq 9} \max\{C_0(y), C_1(y)\}.
\end{align}
Here, $C_0$ increases while $C_1$ decreases, and their relevant pieces intersect uniquely at $y^*=5$, where $C_0(5)=C_1(5)=203$. Evaluating Table~\ref{tab:toy-value-function} at $y^*$ and interpolating its dispatch at the endpoints of the uncertainty set give the following: 
\begin{align}
&y_c^+(\xi)=(5, 5+8\xi, 13),\quad y_s^+(\xi)=(5, 9, 13), \\
&y_p^+(\xi)=(0, 6\xi, 9-7\xi)
\label{eq:toy-recourse-value-function}
\end{align}
The $t+1$ action is independent of $\xi$ and the policy is feasible for every $\xi \in [0,1]$. 

Finally, current dispatch is not imposed. Let $x=y_{s,t} \in [1,9]$. Because $y^*=5$ is ramp-feasible for every such $x$, direct elimination of the remaining variables gives 
\begin{equation}
\begin{aligned}
J_{\mathrm{FAR}}(x)&=
\begin{cases}261-9x,&1\leq x\leq2,\\
253-5x,&2\leq x\leq5,\\203+5x,&5\leq x\leq9,
\end{cases}\\[-1mm]
J_{\mathrm{CAR}}(x)&=
\begin{cases}256-4x,&1\leq x\leq5,\\
221+3x,&5\leq x\leq9.
\end{cases}
\end{aligned}
\label{eq:toy-joint-objectives}
\end{equation}
Both objectives are uniquely minimized at $x=5$, giving total costs 228 and 236. At the CAR optimum, Table~\ref{tab:toy-dispatch} shows that the ramping constraint for $y_{s,t}$ is slack and the capacity constraint for $y_{c,t}$ is binding, so the CAR LMP is the slow generator's marginal cost of $\$4$/MW. Although this price lies within FAR's dual-optimal price interval, CAR does not obtain it by applying an exogenous selection rule to the FAR dual solutions. It is instead the balance multiplier produced by the distinct shared-policy CAR market-clearing problem.

Both formulations implement the same current dispatch but produce different price structures. FAR attains the lower total cost of 228 through uncertainty-dependent recourse. Its active worst-case realizations produce endpoint prices of $\$3$/MW and $\$6$/MW, and convex dual-optimal weightings generate the full price interval $[3,6]$ USD/MW, leaving the current LMP set-valued as characterized in Theorem~\ref{thm:far-current-price}. CAR enforces one causal affine policy, increasing the total cost to 236 while producing a unique LMP of $\$4$/MW, even though multiple uncertainty realizations remain worst-case.

\subsection{Ten-Generator Rolling Market Study}
\label{subsec:ten-generator-study}

The ten-generator study tests whether CAR changes dispatch and prices when uncertainty interacts with binding ramps. It collapses exactly to deterministic LMP at $\Gamma=0$; Fig.~\ref{fig:ten-gen-gamma} and Table~\ref{tab:ten-gen-summary} show increasing effects as ramping tightens.

This study uses a demand trajectory from MISO's projected net load for August 2032, consisting of 288 5-minute intervals in a 24-hour period, which is rescaled to $1050$-MW average (range $636.97$--$1926.28$~MW). Each rolling solve contains a 1-hour look-ahead window. The tight, medium, and loose cases multiply every unit's listed ramp limit by $0.2$, $0.5$, and $1.0$, respectively. The uncertainty calibration uses $\sigma_{\mathrm{rel}}=0.05$, $\rho=0.8$, and $\Gamma\in\{0,1,2,3\}$. Shortage imbalance is penalized at $3500$ USD/MWh. Figures~\ref{fig:ten-gen-gamma} and~\ref{fig:ten-gen-methods} use the 276-window common-state protocol: the first state is a merit-order dispatch at the first load observation, and each subsequent window uses the preceding nominal-LAED dispatch as the common initial state. All models are implemented in Pyomo and solved with Gurobi.

\begin{figure}[t]
    \centering
    \includegraphics[width=\columnwidth]{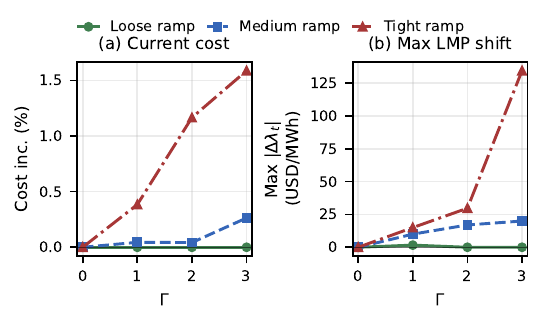}
    \caption{Ten-generator CAR cost and LMP effects across uncertainty budgets and ramp regimes.}
    \label{fig:ten-gen-gamma}
\end{figure}

At $\Gamma=3$, panels (a) and (b) of Fig.~\ref{fig:ten-gen-methods} use all 276 windows. Under loose, medium, and tight ramps, CAR changes current cost by $0.000\%$, $0.264\%$, and $1.587\%$, and its mean signed LMP shifts are $0.000$, $-0.937$, and $-0.804$ USD/MWh, respectively. Panel (c) uses every 24th window, giving 12 comparisons per ramp regime. FAR is solved using the procedure in \cite{lorca2015}, with all $1760$ possible extreme realizations checked. The mean CAR--FAR gaps are $0.97\%$, $1.1\%$, and $49.54\%$. The larger tight-ramp CAR gap and the larger NAR gaps occur because limited ramping and severe uncertainty make future imbalance costly; NAR further restricts future generation by fixing it before uncertainty is known.

\begin{figure*}[t]
    \centering
    \includegraphics[width=\textwidth]{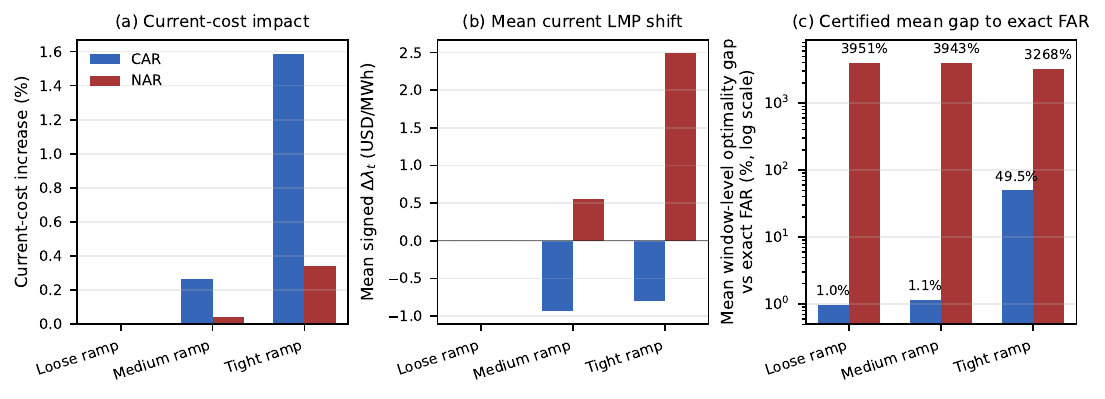}
    \caption{Ten-generator comparison at $\Gamma=3$ across ramping regimes. Measures current-cost impact, LMP shift, and optimality gap.}
    \label{fig:ten-gen-methods}
\end{figure*}

\begin{table}[t]
\centering
\caption{Ten-generator rolling summary relative to deterministic LAED; mean price shifts are signed and maximum shifts are absolute.}
\label{tab:ten-gen-summary}
\footnotesize
\begin{tabular}{lccc}
\hline
Ramp regime & CAR cost effect & Mean $\Delta\lambda_t$ & Max $|\Delta\lambda_t|$ \\
\hline
Loose & $0.000\%$ & $-0.0087$ to $0.000$ & $0.00$--$1.65$ \\
Medium & $0.043$--$0.264\%$ & $-0.937$ to $0.048$ & $10.0$--$20.0$ \\
Tight & $0.384$--$1.587\%$ & $-0.804$ to $-0.599$ & $15.0$--$134.5$ \\
\hline
\end{tabular}
\end{table}

We additionally evaluate the policy in a closed-loop system rather than under the common-state comparison above. This produced no load shedding or ramp slack for $\Gamma=1,2,3$. Its average current-cost increase is $0.12\%$, largest mean price shift is $0.43$ USD/MWh. At tight ramps and $\Gamma=3$, it raises current cost by $0.62\%$, changes the current dispatch in all 48 rolling windows, and moves current output by 336 MW on average in $\ell_1$-norm (194 MW maximum for one unit).

\subsection{Public Network Study}
\label{subsec:public-network-study}

The IEEE 300-bus study converts PGLib-OPF case300 \cite{pglibopf} to a congested DC setting. Its three-future-period runs fix $\Gamma=2$, ramp multiplier of $0.20$, and a $1.15$ line-stress margin, then vary $\sigma_{\mathrm{rel}}$, the per-period parameterization of the uncertainty set with uncertainty only on eight of the largest-load buses with no load shedding throughout. Linear costs are derivatives of the native polynomial costs at the base dispatch; uncertainty is allocated by load over the eight largest-load buses. Figure \ref{fig:pglib300-summary} provides a depiction of the LMP shift and optimality gap across different $\sigma_{\mathrm{rel}}$. 

\begin{figure}[t]
    \centering
    \includegraphics[width=\columnwidth]{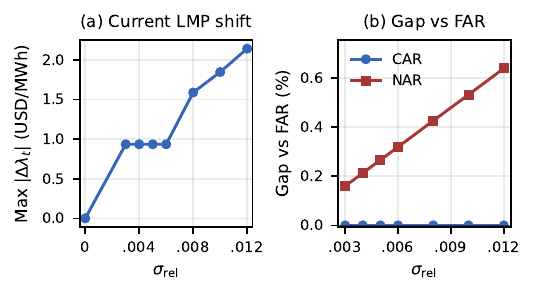}
    \caption{IEEE 300-bus results at $\Gamma=2$: (a) maximum current LMP shift and (b) CAR and NAR objective differences from FAR.}
    \label{fig:pglib300-summary}
\end{figure}

Fig.~\ref{fig:pglib300-summary} shows a maximum LMP shift of $2.139$ USD/MWh and a robust-objective premium over deterministic LAED of $2.004\%$. In panel (b), CAR matches FAR's protected objective at reported precision, whereas the NAR gap reaches $0.64\%$ at $\sigma_{\mathrm{rel}}=0.012$.

Table~\ref{tab:pglib300-summary} reports the objective gap, price shifts, dispatch, and computation time under moderate uncertainty on all buses with $\sigma_{\mathrm{rel}}=0.004$ and $\Gamma=1.5$. The line limit was modified with multiplier $1.10$ and ramp multiplier $0.40$. CAR nearly matches FAR's objective and dispatch, whereas NAR is more conservative and shifts price more on average, however, CAR does have the larger maximum LMP shift.

\begin{table}[t]
\centering
\caption{IEEE 300-bus results with $\sigma_{\mathrm{rel}}=0.004$ and $\Gamma=1.5$. Run times report median [min, max] over 10 runs in seconds.}
\label{tab:pglib300-summary}
\scriptsize
\setlength{\tabcolsep}{1.4pt}
\begin{tabular}{@{}lccccc@{}}
\hline
Method
& Max
& Mean
& Gap
& Dispatch
& Solve time \\
& $|\Delta\lambda|$
& $|\Delta\lambda|$
& vs. FAR
& move
& median [min, max] \\
& \$/MWh
& \$/MWh
& \%
& MW
& sec \\
\hline
Det.
& --
& --
& --
& --
& $0.018\,[0.017,0.018]$ \\
NAR
& $2.380$
& $0.465$
& $0.243$
& $149.8$
& $1.556\,[1.456,1.852]$ \\
CAR
& $2.393$
& $0.398$
& $0.001$
& $140.6$
& $3.601\,[3.283,4.849]$ \\
FAR
& --
& --
& --
& $140.2$
& $68.038\,[66.915,70.114]$ \\
\hline
\end{tabular}
\setlength{\tabcolsep}{6pt}
\end{table}

NAR and CAR are solved directly using their traditional finite support-dual robust-counterpart LPs. FAR is solved by complete vertex enumeration as one globally exact extensive LP containing all 24 extreme uncertainty trajectories. The run times are reported as the Gurobi solve time, excluding the model building time. Together, Figure~\ref{fig:pglib300-summary} and Table~\ref{tab:pglib300-summary} show that CAR provides a practical middle ground between non-adaptive and fully adaptive robust dispatch while producing scenario-independent prices. It preserves much of the operational value of fully adaptive recourse and avoids the greater conservatism of non-adaptive decisions. 

Using the same IEEE 300-bus configuration as Table~\ref{tab:pglib300-summary}, Table~\ref{tab:pglib300-settlement-validation} tests Theorem~\ref{thm:current-incentive}: the bus-level LMP leaves $750.84$ USD of current LOC across four generators, while the forward-ramp adder reduces LOC to zero. The largest absolute adder is $8.62$ USD/MWh. This verifies the theorem's local dispatch-following claim, but does not address any system-wide revenue adequacy.

\begin{table}[t]
\centering
\caption{IEEE 300-bus LOC validation at $\sigma_{\mathrm{rel}}=0.004$ and $\Gamma=1.5$.}
\label{tab:pglib300-settlement-validation}
\scriptsize
\begin{tabular}{lccc}
\hline
Settlement rule & Total LOC & Units with LOC & Supported \\
\hline
CAR LMP only & $\$750.84$ & $4$ & No \\
CAR LMP + ramp adder & $0.00$ & $0$ & Yes \\
\hline
\end{tabular}
\end{table}

\section{Conclusion}
\label{sec:conclusion}

This paper addresses a fundamental price-formation problem in adaptive robust look-ahead economic dispatch. Although fully adaptive robust LAED can produce a reliable and cost-efficient dispatch, multiple active worst-case trajectories can imply different marginal values of demand, so the formulation does not always endogenously select a unique, scenario-independent market price. CAR resolves this ambiguity by restricting future recourse to a shared causal affine policy. The resulting convex robust counterpart provides explicit nominal-balance multipliers and, subject to the usual dual-uniqueness qualification, a standard DC LMP decomposition. The associated ramp-adjusted settlement also supports current-period dispatch-following and eliminates current-period lost opportunity cost under the stated assumptions.

The analytical and computational results clarify both the value and the limitations of this restriction. CAR prices collapse to deterministic LAED prices when robustness is inactive or flexibility is abundant, but change when uncertainty interacts with scarce ramping capability or transmission congestion. Across the ten-generator and IEEE 300-bus studies, CAR avoids much of the conservatism of NAR while retaining much of the operational value of fully adaptive recourse and remaining computationally practical. At the same time, the larger CAR--FAR gaps under tight ramping and high uncertainty show that affine recourse is not universally exact and that higher uncertainty can widen the optimality gap. CAR is therefore a price-forming bridge between deterministic market clearing and fully adaptive robust operation.

The present analysis is limited to marginal pricing for economic dispatch under DC power flow, affine recourse decisions, and current-period incentive compatibility. Future work should investigate richer causal policy classes, including piecewise-affine or polynomial rules, that may reduce the CAR--FAR gap while preserving transparent price formation. Additional priorities include multi-period settlement designs that address revenue adequacy, uplift allocation, and strategic bidding. Further work should also extend the framework to AC network constraints and analyze the relationship between adaptive robust dispatch and reserve procurement.

\bibliographystyle{IEEEtran}
\bibliography{references}

\newpage

\onecolumn

\appendices

\section{Full CAR Robust Counterpart Derivation}
\label{app:affine-counterpart}

This appendix gives the explicit causal affine recourse (CAR) reformulation in terms of the support function of the uncertainty set. Keeping the reformulation in support-function form allows different uncertainty sets to be used without changing the dispatch-side reformulation, as described in \cite{bertsimas2022robust}. The notation follows the DC network formulation in Section~\ref{sec:car_marginal_pricing}. We use an affine approximation of the fully adaptive robust problem for decision variables in future intervals by enforcing
\[
x_{g,t+\tau}(\xi)=x^0_{g,t+\tau}+R_{g,t+\tau}\xi
\]
and
\[
s_{t+\tau}(\xi)=s^0_{t+\tau}+S_{t+\tau}\xi,
\]
where $x^0_{g,t+\tau} \in \mathbb{R}$ and $s^0_{t+\tau} \in \mathbb{R}_+^{|\mathcal{N}|}$ are nominal future decisions, while $R_{t+\tau} \in \mathbb{R}^{|\mathcal{G}| \times m}$ and $S_{t+\tau} \in \mathbb{R}^{|\mathcal{N}| \times m}$ are linear recourse matrices. These matrices are constrained causally so that the decisions at time $t+\tau$ depend only on components of $\xi$ revealed by time $t+\tau$. We express this as
\[
(R,S) \in \mathcal{Y}_{\rm causal}.
\]

Here, \(x^0_t \in \mathbb{R}^{|\mathcal G|}\) and \(s^0_t \in \mathbb{R}^{|\mathcal N|}\) are the current-period decisions, \(G \in \mathbb{R}^{|\mathcal N|\times |\mathcal G|}\) is the generator-to-bus incidence matrix, \(\Pi_\ell\) denotes the \(\ell\)-th row of the PTDF matrix, and \(D_{t+\tau}\in\mathbb{R}^{|\mathcal N|\times m}\) maps the primitive uncertainty vector \(\xi\) to the nodal net-load deviation at time \(t+\tau\). We take \(\overline r_g\ge 0\) and \(\underline r_g\ge 0\) to denote the magnitudes of the upward and downward ramp limits, respectively. All vector inequalities are interpreted component-wise. The uncertainty set \(\Xi\) is assumed to be nonempty and compact and to contain the nominal realization \(\xi=0\). The affinely adaptive robust dispatch problem is:
{\footnotesize
\begin{subequations}\label{eq:app-car-semi-infinite}
\begin{align}
\min_{x^0,s^0,R,S,\eta}\quad
& c^\top x^0_t + C^{\rm shed}\mathbf{1}^\top s^0_t + \eta
\label{eq:app-car-obj}\\
\text{s.t.}\quad
& \eta \ge
\sum_{\tau\in\mathcal{W}^+}
\left[
c^\top\left(x^0_{t+\tau}+R_{t+\tau}\xi\right)
+
C^{\rm shed}\mathbf{1}^\top
\left(s^0_{t+\tau}+S_{t+\tau}\xi\right)
\right],
&& \forall \xi\in\Xi,
\label{eq:app-car-eta}\\
& \mathbf{1}^\top
\left(
G x^0_t+s^0_t-\hat d_t
\right)=0,
\label{eq:app-car-current-balance}\\
& \mathbf{1}^\top
\left[
G\left(x^0_{t+\tau}+R_{t+\tau}\xi\right)
+
\left(s^0_{t+\tau}+S_{t+\tau}\xi\right)
-
\left(\hat d_{t+\tau}+D_{t+\tau}\xi\right)
\right]=0,
&& \forall \tau\in\mathcal{W}^+,\ \forall \xi\in\Xi,
\label{eq:app-car-future-balance}\\
& -\overline f_{\ell,t}
\le
\Pi_\ell\left(G x^0_t+s^0_t-\hat d_t\right)
\le
\overline f_{\ell,t},
&& \forall \ell\in\mathcal{L},
\label{eq:app-car-current-flow}\\
& -\overline f_{\ell,t+\tau}
\le
\Pi_\ell
\left[
G\left(x^0_{t+\tau}+R_{t+\tau}\xi\right)
+
\left(s^0_{t+\tau}+S_{t+\tau}\xi\right)
-
\left(\hat d_{t+\tau}+D_{t+\tau}\xi\right)
\right]
\le
\overline f_{\ell,t+\tau},
&& \forall \ell\in\mathcal{L},\
   \forall \tau\in\mathcal{W}^+,\
   \forall \xi\in\Xi,
\label{eq:app-car-future-flow}\\
& \underline x_g
\le x^0_{g,t}
\le \overline x_g,
&& \forall g\in\mathcal{G},
\label{eq:app-car-current-cap}\\
& \underline x_g
\le x^0_{g,t+\tau}+R_{g,t+\tau}\xi
\le \overline x_g,
&& \forall g\in\mathcal{G},\
   \forall \tau\in\mathcal{W}^+,\
   \forall \xi\in\Xi,
\label{eq:app-car-future-cap}\\
& x^0_{g,t}-x^{\rm impl}_{g,t-1}\le \overline{r}_g,
&& \forall g\in\mathcal{G},
\label{eq:app-car-current-ramp-up}\\
& x^{\rm impl}_{g,t-1}-x^0_{g,t}\le \underline{r}_g,
&& \forall g\in\mathcal{G},
\label{eq:app-car-current-ramp-down}\\
& x^0_{g,t+1}+R_{g,t+1}\xi - x^0_{g,t}\le \overline{r}_g,
&& \forall g\in\mathcal{G},\ \forall \xi\in\Xi,
\label{eq:app-car-first-ramp-up}\\
& x^0_{g,t}-x^0_{g,t+1}-R_{g,t+1}\xi\le \underline{r}_g,
&& \forall g\in\mathcal{G},\ \forall \xi\in\Xi,
\label{eq:app-car-first-ramp-down}\\
& x^0_{g,t+\tau}+R_{g,t+\tau}\xi
-
x^0_{g,t+\tau-1}-R_{g,t+\tau-1}\xi
\le \overline{r}_g,
&& \forall g\in\mathcal{G},\
   \forall \tau\in\{2,\ldots,H\},\
   \forall \xi\in\Xi,
\label{eq:app-car-future-ramp-up}\\
& x^0_{g,t+\tau-1}+R_{g,t+\tau-1}\xi
-
x^0_{g,t+\tau}-R_{g,t+\tau}\xi
\le \underline{r}_g,
&& \forall g\in\mathcal{G},\
   \forall \tau\in\{2,\ldots,H\},\
   \forall \xi\in\Xi,
\label{eq:app-car-future-ramp-down}\\
& s^0_t\ge 0,
\label{eq:app-car-current-shortage}\\
& s^0_{n,t+\tau}+S_{n,t+\tau}\xi\ge 0,
&& \forall n\in\mathcal{N},\
   \forall \tau\in\mathcal{W}^+,\
   \forall \xi\in\Xi,
\label{eq:app-car-future-shortage}\\
& (R,S)\in\mathcal{Y}_{\rm causal}.
\label{eq:app-car-causality}
\end{align}
\end{subequations}
}

Define the support function of $\Xi$ by
\[
    \sigma_\Xi(a):=\max_{\xi\in\Xi} a^\top \xi .
\]
The future power-balance constraint is an affine equality in $\xi$:
\[
    \mathbf{1}^{\top}
    \left[
        G\left(x^0_{t+\tau}+R_{t+\tau}\xi\right)
        +\left(s^0_{t+\tau}+S_{t+\tau}\xi\right)
        -\left(\hat d_{t+\tau}+D_{t+\tau}\xi\right)
    \right]=0,
    \qquad \forall \xi\in\Xi .
\]
For coefficient matching, we represent the uncertainty in coordinates spanning the affine hull of \(\Xi\), so that \(\Xi\) is full dimensional in the chosen uncertainty coordinates. Under this convention, the robust equality is equivalent to
\begin{subequations}\label{eq:app-car-balance-matching}
\begin{align}
    \mathbf{1}^\top
    \left(Gx^0_{t+\tau}+s^0_{t+\tau}-\hat d_{t+\tau}\right)
    &=0,
    && \forall \tau\in\mathcal{W}^+,
    \label{eq:app-car-nominal-balance}\\
    \mathbf{1}^\top
    \left(GR_{t+\tau}+S_{t+\tau}-D_{t+\tau}\right)
    &=0,
    && \forall \tau\in\mathcal{W}^+.
    \label{eq:app-car-affine-balance}
\end{align}
\end{subequations}
The first equality is the nominal system-balance equation, while the second requires the affine generation and shortage response to track the uncertainty-induced change in total net load. The degenerate case \(\Xi=\{0\}\) reduces to deterministic LAED, in which only the nominal balance equation is required.

The tractable robust counterpart of \eqref{eq:app-car-semi-infinite} is therefore the following in terms of the support function:

\begin{subequations}
\label{eq:app-car-robust-counterpart}
\begin{align}
\min_{x^0,s^0,R,S,\eta}\quad
& c^\top x^0_t + C^{\rm shed}\mathbf{1}^\top s^0_t + \eta
\label{eq:app-rc-obj}\\
\text{s.t.}\quad
& \sum_{\tau\in\mathcal{W}^+}
\left(c^\top x^0_{t+\tau}
+
C^{\rm shed}\mathbf{1}^\top s^0_{t+\tau}\right)
+
\sigma_\Xi\!\left(
\sum_{\tau\in\mathcal{W}^+}
\left(
R_{t+\tau}^\top c
+
C^{\rm shed} S_{t+\tau}^\top \mathbf{1}
\right)
\right)
\le \eta,
\label{eq:app-rc-eta}\\
& \mathbf{1}^\top
\left(Gx^0_t+s^0_t-\hat d_t\right)=0,
\label{eq:app-rc-current-balance}\\
& \mathbf{1}^\top
\left(Gx^0_{t+\tau}+s^0_{t+\tau}-\hat d_{t+\tau}\right)=0,
&& \forall \tau\in\mathcal{W}^+,
\label{eq:app-rc-future-balance-nom}\\
& \mathbf{1}^\top
\left(GR_{t+\tau}+S_{t+\tau}-D_{t+\tau}\right)=0,
&& \forall \tau\in\mathcal{W}^+,
\label{eq:app-rc-future-balance-aff}\\
& \Pi_\ell
\left(Gx^0_t+s^0_t-\hat d_t\right)
\le \overline f_{\ell,t},
&& \forall \ell\in\mathcal{L},
\label{eq:app-rc-current-flow-up}\\
& -\Pi_\ell
\left(Gx^0_t+s^0_t-\hat d_t\right)
\le \overline f_{\ell,t},
&& \forall \ell\in\mathcal{L},
\label{eq:app-rc-current-flow-down}\\
& \Pi_\ell
\left(Gx^0_{t+\tau}+s^0_{t+\tau}-\hat d_{t+\tau}\right)
+
\sigma_\Xi\!\left(
\left(GR_{t+\tau}+S_{t+\tau}-D_{t+\tau}\right)^\top
\Pi_\ell^\top
\right)
\le \overline f_{\ell,t+\tau},
&& \forall \ell\in\mathcal{L},\
   \forall \tau\in\mathcal{W}^+,
\label{eq:app-rc-future-flow-up}\\
& -\Pi_\ell
\left(Gx^0_{t+\tau}+s^0_{t+\tau}-\hat d_{t+\tau}\right)
+
\sigma_\Xi\!\left(
-\left(GR_{t+\tau}+S_{t+\tau}-D_{t+\tau}\right)^\top
\Pi_\ell^\top
\right)
\le \overline f_{\ell,t+\tau},
&& \forall \ell\in\mathcal{L},\
   \forall \tau\in\mathcal{W}^+,
\label{eq:app-rc-future-flow-down}\\
& \underline x_g
\le x^0_{g,t}
\le \overline x_g,
&& \forall g\in\mathcal{G},
\label{eq:app-rc-current-cap}\\
& x^0_{g,t+\tau}+\sigma_\Xi(R_{g,t+\tau}^\top)
\le \overline x_g,
&& \forall g\in\mathcal{G},\
   \forall \tau\in\mathcal{W}^+,
\label{eq:app-rc-future-cap-up}\\
& \underline x_g+\sigma_\Xi(-R_{g,t+\tau}^\top)
\le x^0_{g,t+\tau},
&& \forall g\in\mathcal{G},\
   \forall \tau\in\mathcal{W}^+,
\label{eq:app-rc-future-cap-down}\\
& x^0_{g,t}-x^{\rm impl}_{g,t-1}\le \overline{r}_g,
&& \forall g\in\mathcal{G},
\label{eq:app-rc-current-ramp-up}\\
& x^{\rm impl}_{g,t-1}-x^0_{g,t}\le \underline{r}_g,
&& \forall g\in\mathcal{G},
\label{eq:app-rc-current-ramp-down}\\
& x^0_{g,t+1}-x^0_{g,t}
+
\sigma_\Xi(R_{g,t+1}^\top)
\le \overline{r}_g,
&& \forall g\in\mathcal{G},
\label{eq:app-rc-first-ramp-up}\\
& x^0_{g,t}-x^0_{g,t+1}
+
\sigma_\Xi(-R_{g,t+1}^\top)
\le \underline{r}_g,
&& \forall g\in\mathcal{G},
\label{eq:app-rc-first-ramp-down}\\
& x^0_{g,t+\tau}-x^0_{g,t+\tau-1}
+
\sigma_\Xi\!\left(
(R_{g,t+\tau}-R_{g,t+\tau-1})^\top
\right)
\le \overline{r}_g,
&& \forall g\in\mathcal{G},\
   \forall \tau\in\{2,\ldots,H\},
\label{eq:app-rc-future-ramp-up}\\
& x^0_{g,t+\tau-1}-x^0_{g,t+\tau}
+
\sigma_\Xi\!\left(
(R_{g,t+\tau-1}-R_{g,t+\tau})^\top
\right)
\le \underline{r}_g,
&& \forall g\in\mathcal{G},\
   \forall \tau\in\{2,\ldots,H\},
\label{eq:app-rc-future-ramp-down}\\
& s^0_t\ge 0,
\label{eq:app-rc-current-shortage}\\
& \sigma_\Xi(-S_{n,t+\tau}^\top)\le s^0_{n,t+\tau},
&& \forall n\in\mathcal{N},\
   \forall \tau\in\mathcal{W}^+,
\label{eq:app-rc-future-shortage}\\
& (R,S)\in\mathcal{Y}_{\rm causal}.
\label{eq:app-rc-causality}
\end{align}
\end{subequations}

For a nonempty compact polyhedral uncertainty set
\[
    \Xi=\{\xi:H_\Xi \xi\le h_\Xi\},
\]
each support-function term can be replaced by the linear dual representation
\[
    \sigma_\Xi(a)
    =
    \min_{\lambda\ge 0}
    \left\{
        h_\Xi^\top \lambda:
        H_\Xi^\top \lambda=a
    \right\}.
\]
A separate dual vector \(\lambda\) is introduced for each occurrence of the support function. The auxiliary innovation variables used to define the dynamic-budget uncertainty set need not be projected out explicitly; equivalently, the support-function maximization can be dualized directly in the corresponding lifted polyhedral representation. Substituting the appropriate dual representation for every occurrence of $\sigma_\Xi(\cdot)$ in \eqref{eq:app-car-robust-counterpart} gives a finite-dimensional linear program whenever the cost and network model are linear. This is the tractable robust counterpart used for the CAR computational experiments.

\end{document}